\documentclass[dvipsnames,letterpaper, 10 pt, conference]{ieeeconf}

\IEEEoverridecommandlockouts	

\usepackage{cite}
\usepackage{amsmath,amssymb,amsfonts}
\usepackage{algorithmic}
\usepackage{textcomp}
\usepackage{graphicx,color}
\usepackage{mathrsfs}
\usepackage[vlined,ruled]{algorithm2e}
\usepackage{subfigure}
\usepackage{url}
\usepackage{color}
\usepackage{dsfont}
\usepackage{bbm}
\usepackage{booktabs}
\usepackage{array}
\usepackage[table]{xcolor} 
\usepackage{yfonts}
\usepackage[normalem]{ulem}
\usepackage{arydshln}
\usepackage{multicol}
\usepackage{amsmath}
\usepackage{stfloats}

  \newenvironment{smallarray}[1]
 {\null\,\vcenter\bgroup\scriptsize
  \arraycolsep=.13885em
  \hbox\bgroup$\array{@{}#1@{}}}
 {\endarray$\egroup\egroup\,\null}

\newtheorem{theorem}{Theorem}[section]
\newtheorem{lemma}[theorem]{Lemma}

\newtheorem{example}{Example}
\newtheorem{appxlem}{Lemma}[section]

\newcommand{\subscr}[2]{{#1}_{\textup{#2}}}

\usepackage{tabularx}
\usepackage{multirow}
\usepackage{makecell}

\newcommand\aamsout{\bgroup\markoverwith{\textcolor{violet}{\rule[0.5ex]{2pt}{1pt}}}\ULon}

\newcommand{\Rank}{\operatorname{Rank}}

\newcommand{\real}{\mathbb{R}}

\newcommand{\transpose}{\mathsf{T}} 
\newcommand{\T}{\mathsf{T}} 

\newcommand{\mc}{\mathcal}

\newcommand{\expect}[1]{\mathbb{E}\left[#1\right]}
\newcommand{\Tr}[1]{\mathrm{tr}\left[#1\right]}

\newcommand{\overbar}[1]{\mkern 1.5mu\overline{\mkern-1.5mu#1\mkern-1.5mu}\mkern 1.5mu}

\DeclareSymbolFont{bbold}{U}{bbold}{m}{n}
\DeclareSymbolFontAlphabet{\mathbbold}{bbold}

\newcommand\oprocendsymbol{\hbox{$\square$}}
\newcommand\oprocend{\relax\ifmmode\else\unskip\hfill\fi\oprocendsymbol}

\renewcommand{\theenumi}{(\roman{enumi})}

\newcommand*{\QEDA}{\hfill\ensuremath{\blacksquare}}%
\newenvironment{pfof}[1]{\vspace{1ex}\noindent{\itshape Proof of
    #1:}\hspace{0.5em}} {\hfill\QEDA\vspace{1ex}}

\graphicspath{{figs/}}

\makeatletter
\let\NAT@parse\undefined
\makeatother
\usepackage[colorlinks,urlcolor=blue, linkcolor=blue, citecolor=blue]{hyperref}

\newcounter{MYtempeqncnt}

\begin{document}

\title{\LARGE \bf Behavioral Feedback for Optimal LQG Control}


\author{Abed~AlRahman~Al~Makdah,~Vishaal~Krishnan,~Vaibhav~Katewa,~and~Fabio~Pasqualetti
  \thanks{This material is based upon work supported in part by awards
    ONR-N00014-19-1-2264 and AFOSR-FA9550-20-1-0140.  A. A. Al Makdah
    and F. Pasqualetti are with the Department of Electrical and
    Computer Engineering and the Department of Mechanical Engineering
    at the University of California, Riverside, respectively,
    \href{mailto:aalmakdah@engr.ucr.edu}{\{\texttt{aalmakdah}},\href{mailto:fabiopas@engr.ucr.edu}{\texttt{fabiopas\}@engr.ucr.edu}}.
    V. Krishnan is with the School of Engineering and Applied Sciences, Harvard University, \href{mailto:vkrishnan@seas.harvard.edu}{\texttt{vkrishnan@seas.harvard.edu}}.  V. Katewa is
    with the Department of Electrical Communication Engineering at the
    Indian Institute of Science, Bangalore, India,
    \href{mailto:vkatewa@iisc.ac.in}{\texttt{vkatewa@iisc.ac.in}}.}}

\maketitle

\begin{abstract}
In this work, we revisit the Linear Quadratic Gaussian (LQG) optimal control problem from a behavioral perspective. 
Motivated by the suitability of behavioral models for data-driven control, 
we begin with a reformulation of the LQG problem in the space of input-output behaviors 
and obtain a complete characterization of the optimal solutions.
In particular, we show that the optimal LQG controller can be expressed as a static behavioral-feedback gain,
thereby eliminating the need for dynamic state estimation characteristic of state space methods. 
The static form of the optimal LQG gain also makes it amenable to its computation by 
gradient descent, which we investigate via numerical experiments. 
Furthermore, we highlight the advantage of this approach in the data-driven
control setting of learning the optimal~LQG~controller\mbox{~from~expert~demonstrations.}
\end{abstract}

%

\section{Introduction}\label{sec: introduction}
Data-driven control has received increasing interest during the past few years. Specifically, this interest has been surging towards optimal control problems\cite{GB-VK-FP:19,JC-JL-FD:18,ST-BR:19,FC-GB-FP:21}. The Linear Quadratic Gaussian (LQG) is one of the most fundamental optimal control problems, which deals with partially-observed linear dynamical systems in the presence of additive white gaussian noise \cite{KZ-JCD-KG:96}. When the system is known, the LQG problem enjoys an elegant closed-form solution obtained via the separation principle \cite[Theorem 14.7]{KZ-JCD-KG:96}. In the context of data-driven control, however, the LQG problem is less studied in the literature due to some major challenges: (i) the states of the system cannot be directly measured for learning purposes, (ii) the optimal policy is expressed in the dynamic controller form where it is not unique \cite{KZ-JCD-KG:96}, and (iii) the set of stabilizing controllers can be disconnected \cite{YZ-YT-NL:21}. On the other hand, the Linear Quadratic Regulator (LQR) optimal control problem has received more attention in the context of data-driven control \cite{GB-VK-FP:19,ST-BR:19,FC-GB-FP:21}. Some of the reasons that make the LQR problem attractive is that the optimal policy can be expressed as a static feedback gain and it is unique \cite[Theorem 14.2]{KZ-JCD-KG:96}. Moreover, the set of stabilizing feedback gains for the LQR problem is connected and the LQR cost function is gradient dominant \cite{MF-RG-SK-MM:18,HM-AZ-MS-MRJ:19}. These properties are useful for providing convergence guarantees for gradient-based methods for solving the LQR problem \cite{JB-AM-MF-MM:19,IF-BP:21} as well as for model-free policy optimization methods \cite{HM-AZ-MS-MRJ:21}. However, LQR controllers require measuring the full states, and are used in deterministic settings, which limits the use of LQR controllers in practical control applications.\\
In this paper, we make the LQG problem more accessible for data-driven methods. In particular, we show that the optimal LQG controller can be expressed as a static feedback gain by reformulation of the model-based LQG problem in the space of input-output behaviors. Then, we highlight the advantages of having a static LQG gain in the context of data-driven control and gradient-based~algorithms.\\
\textbf{Related work.} The LQG control problem has been studied extensively in the literature \cite{KZ-JCD-KG:96,DPB:01a}, where fundamental properties have been characterized, such as the existence of optimal solution, how to obtain it using separation principle \cite{KZ-JCD-KG:96}, and its lack of stability margin guarantees in closed-loop \cite{JCD:78}. 
However, in the context of data-driven control, the LQR problem has received more attention than the LQG problem. 
The landscape properties for the LQR problem with state-feedback control has been studied in \cite{MF-RG-SK-MM:18,HM-AZ-MS-MRJ:19}, 
which has paved the way for subsequent works investigating convergence properties of gradient methods for solving the LQR problem~\cite{JB-AM-MF-MM:19,IF-BP:21,HM-AZ-MS-MRJ:21}. 
Recent studies have revisited the LQG problem in the context of data-driven methods 
(e.g. \cite{SL-KA-BH-AA:20,LF-YZ-MK:20,YZ-LF-MK-NL:21}). In \cite{YZ-YT-NL:21}, the authors characterize the optimization landscape for the LQG problem, showing that the set of stabilizing dynamic controllers can be disconnected. In the context of data-driven control, the behavioral approach has 
garnered much attention in recent years~\cite{JCW-PR-IM-BLMDM:05,CDP-PT:19,LF-BG-AM-GFT:21,VK-FP:21},
as it circumvents the need for state space representation.
Owing it this fact, it belongs in the same category as the
difference operator representation and ARMAX models~\cite[Sec. 2.3 and Sec. 7.4]{GCG-KSS:14},
and shares several connections with these classes of models. We refer the reader to \cite{IM-FD:21-survey} for a comprehensive overview of the behavioral approach.\\
Despite recent interest in the behavioral approach, a fundamental
understanding of the LQG problem from a behavioral perspective 
is still lacking, and our work addresses this gap. 
Different from the literature, our work seeks to characterize the optimal behavioral feedback
controllers for the LQG problem in model-based setting, and to demonstrate their suitability for
data-driven control and gradient-based methods for controller design.
More specifically, we show that the optimal LQG controller can be expressed as a 
static behavioral-feedback gain, which underlies its advantages
for developing data-driven methods to learn LQG controllers.

%
%
%
\textbf{Contributions.} This paper features three main contributions. First, we introduce equivalent representations for stochastic discrete-time, linear, time-invariant systems and the LQG optimal control problem in the behavioral space (Lemma~\ref{lemma: system in z} and Lemma~\ref{lemma: lqg in z}, respectively). Second, we show that, in the behavioral space, the optimal LQG controller can be expressed as a static behavioral-feedback gain, which can be solved for directly from the LQG problem represented in the behavioral space (Theorem \ref{thrm: solution of lqg in z}). Third, we highlight the advantages of having a static feedback LQG gain over a dynamic LQG controller in the context of data-driven control and gradient-based algorithms (section~\ref{sec: numerical methods}).\\
%
\textbf{Notation.} A Gaussian random variable $x$ with mean $\mu$ and
covariance $\Sigma$ is denoted as $x\sim\mc{N}(\mu,\Sigma)$. The
$n\times n$ identity matrix is denoted by $I_n$. The expectation
operator is denoted by $\mathbb{E}[\cdot]$. The spectral radius and
the trace of a square matrix $A$ are denoted by $\rho(A)$ and
$\Tr{A}$, respectively. A positive definite (semidefinite) matrix
$A$ is denoted as $A\succ 0$ ($A\succeq 0$). The Kronecker product is denoted
by $\otimes$, and vectorization operator is denoted by vec($\cdot$). The left (right) pseudo inverse of a tall (fat) matrix $A$ is denoted by $A^{\dagger}$.
\begin{figure*}[!b]
\normalsize
\setcounter{MYtempeqncnt}{0}
\setcounter{equation}{4}
\hrulefill 
\smallskip
\begin{align} \label{eq: z dynamics}
\begin{aligned}
\underbrace{\left[ \begin{smallarray}{c}
u(t-n+1) \\
\vdots \\
u(t-1)\\
u(t) \\ 
\hdashline[2pt/2pt] 
y(t-n+1)\\
\vdots \\
y(t)\\
y(t+1) \\ 
\hdashline[2pt/0pt] 
w(t-n+1) \\
\vdots \\
w(t-1) \\
w(t) \\ 
\hdashline[2pt/2pt] 
v(t-n+1)\\
\vdots \\
v(t)\\
v(t+1)
\end{smallarray} \right]}_{z(t+1)}=&
\underbrace{\left[ \begin{smallarray}{ccccc;{2pt/2pt}ccccc;{2pt/0pt}ccccc;{2pt/2pt}ccccc}
0 & I & 0 & \cdots & 0 & 0 & 0 & 0& \cdots & 0 & 0 & 0 & 0& \cdots & 0 & 0 & 0 & 0& \cdots & 0\\
\vdots & \vdots & \vdots & \ddots & \vdots & \vdots & \vdots & \vdots & \ddots & \vdots & \vdots & \vdots & \vdots & \ddots & \vdots & \vdots & \vdots & \vdots & \ddots & \vdots \\
0 & 0 & 0 & \cdots & I & 0 & 0 & 0& \cdots & 0 & 0 & 0 & 0& \cdots & 0 & 0 & 0 & 0& \cdots & 0\\
0 & 0 & 0 & \cdots & 0 & 0 & 0 & 0& \cdots & 0 & 0 & 0 & 0& \cdots & 0 & 0 & 0 & 0& \cdots & 0\\
\hdashline[2pt/2pt] 
0 & 0 & 0 & \cdots & 0 & 0 & I & 0 & \cdots & 0 & 0 & 0 & 0& \cdots & 0 & 0 & 0 & 0& \cdots & 0\\
\vdots & \vdots & \vdots & \ddots & \vdots & \vdots & \vdots & \vdots & \ddots & \vdots & \vdots & \vdots & \vdots & \ddots & \vdots & \vdots & \vdots & \vdots & \ddots & \vdots \\
0 & 0 & 0& \cdots  & 0 & 0 & 0 & 0 & \cdots & I & 0 & 0 & 0& \cdots & 0 & 0 & 0 & 0& \cdots & 0\\
\multicolumn{5}{c;{2pt/2pt}}{\mc{A}_u} & \multicolumn{5}{c;{2pt/0pt}}{\mc{A}_y} & \multicolumn{5}{c;{2pt/2pt}}{\mc{A}_w} & \multicolumn{5}{c}{\mc{A}_v}\\
\hdashline[2pt/0pt] 
0 & 0 & 0 & \cdots & 0 & 0 & 0 & 0 & \cdots & 0 & 0 & I & 0 & \cdots & 0 & 0 & 0 & 0& \cdots & 0\\
\vdots & \vdots & \vdots & \ddots & \vdots & \vdots & \vdots & \vdots & \ddots & \vdots & \vdots & \vdots & \vdots & \ddots & \vdots & \vdots & \vdots & \vdots & \ddots & \vdots \\
0 & 0 & 0& \cdots  & 0 & 0 & 0 & 0 & \cdots & 0 & 0 & 0 & 0 & \cdots & I & 0 & 0 & 0& \cdots & 0 \\
0 & 0 & 0& \cdots  & 0 & 0 & 0 & 0 & \cdots & 0 & 0 & 0 & 0 & \cdots & 0 & 0 & 0 & 0& \cdots & 0 \\
\hdashline[2pt/2pt] 
0 & 0 & 0 & \cdots & 0 & 0 & 0 & 0 & \cdots & 0 & 0 & 0 & 0 & \cdots & 0 & 0 & I & 0 & \cdots & 0\\
\vdots & \vdots & \vdots & \ddots & \vdots & \vdots & \vdots & \vdots & \ddots & \vdots & \vdots & \vdots & \vdots & \ddots & \vdots & \vdots & \vdots & \vdots & \ddots & \vdots\\
0 & 0 & 0 & \cdots & 0 & 0 & 0 & 0& \cdots  & 0 & 0 & 0 & 0 & \cdots & 0 & 0 & 0 & 0 & \cdots & I \\
0 & 0 & 0 & \cdots & 0 & 0 & 0 & 0& \cdots  & 0 & 0 & 0 & 0 & \cdots & 0 & 0 & 0 & 0 & \cdots & 0
\end{smallarray} \right]}_{\mc{A}}
\underbrace{\left[ \begin{smallarray}{c}
u(t-n) \\
\vdots \\
u(t-2) \\
u(t-1) \\ 
\hdashline[2pt/2pt] 
y(t-n)\\
\vdots \\
y(t-1)\\
y(t) \\
\hdashline[2pt/0pt] 
w(t-n) \\
\vdots \\
w(t-2) \\
w(t-1) \\ 
\hdashline[2pt/2pt] 
v(t-n)\\
\vdots \\
v(t-1)\\
v(t)
\end{smallarray} \right]}_{z(t)}
+\underbrace{\left[ \begin{smallarray}{c;{2pt/0pt}c;{2pt/0pt}c}
0 & 0 & 0\\
\vdots & \vdots & \vdots \\
0 & 0 & 0\\
I  & 0 & 0 \\
\hdashline[2pt/2pt] 
0 & 0 & 0\\
\vdots & \vdots & \vdots \\
0 & 0 & 0\\
CB & C & I \\ 
\hdashline[2pt/0pt] 
0 & 0 & 0\\
\vdots & \vdots & \vdots \\
0 & 0 & 0\\
0 & I & 0\\ 
\hdashline[2pt/2pt] 
0 & 0 & 0\\
\vdots & \vdots & \vdots \\
0 & 0 & 0\\
0 & 0 & I
\end{smallarray} \right]}_{\left[ \begin{smallarray}{c;{1pt/0pt}c;{1pt/0pt}c}
\mc{B}_u & \mc{B}_w & \mc{B}_v
\end{smallarray}\right]}
\left[ \begin{smallarray}{c}
u(t)\\
w(t)\\
v(t+1)
\end{smallarray} \right],\\
y_z(t)=&\underbrace{\left[\begin{smallarray}{cccc;{2pt/0pt}cccc}
I & 0 & \cdots & 0 & 0 & 0 & \cdots & 0\\
0 & I & \cdots & 0 & 0 & 0 & \cdots & 0\\
\vdots & \vdots & \ddots & \vdots & \vdots & \vdots & \ddots & \vdots \\
0 & 0 & \cdots & I & 0 & 0 & \cdots & 0
\end{smallarray}\right]}_{\mc{C}} z(t) 
\end{aligned}
\end{align}
\setcounter{equation}{\value{MYtempeqncnt}}
\end{figure*}
\section{Problem setup and main results}\label{sec: setup}
Consider the discrete-time, linear, time-invariant system
\begin{align}
\begin{aligned}\label{eq: system in x}
  x(t+1) &= Ax(t) + Bu(t)+w(t),  \\ 
  y(t) &= Cx(t) + v(t), \qquad t\geq 0,
 \end{aligned}
\end{align}
where $x(t)\in\real^{n}$ denotes the state, $u(t)\in \real^{m}$ the
control input, $y(t)\in\real^{p}$ the measured output, $w(t)$ the
process noise, and $v(t)$ the measurement noise at time $t$. We assume that
$w(t)\sim\mc{N}(0,Q_w)$, with $Q_w\succeq 0$, $v(t)\sim\mc{N}(0,R_v)$,
with $R_v\succ 0$, and $x(0)\sim\mc{N}(0,\Sigma_0)$, with
$\Sigma_0 \succeq 0$, are independent of each other at all
times. For the system \eqref{eq: system in x}, the Linear Quadratic Gaussian
(LQG) problem asks to find a control input that minimizes the cost
\begin{align}\label{eq: LQG cost}
  \mc{J}\triangleq\lim_{T\rightarrow \infty}\mathbb{E} \left[
  \frac{1}{T}\Big(\sum_{t=0}^{T-1}x(t)^{\T}Q_x x(t) +
  u(t)^{\T} R_u u(t)\Big) \right] , 
\end{align}
where $Q_x \succeq 0$ and $R_u \succ 0$ are weighing matrices of~appropriate dimension. We assume that $(A,B)$ and $(A,Q_w^{1/2})$ are controllable, and $(A,C)$ and $(A,Q_x^{1/2})$ are observable. As a classic result \cite{KZ-JCD-KG:96}, the optimal control input that solves the LQG problem can be generated by a dynamic controller of the form
\begin{align}\label{eq: compensator}
\begin{aligned}
x_c(t+1) &= Ex_c(t)+Fy(t), \quad t\geq 0, \\
u(t) &= Gx_c(t) + Hy(t),
\end{aligned}
\end{align}
where $x_c(t)$ denotes the state at time $t$, and $E\in \mathbb{R}^{n\times n}$,
$F\in \mathbb{R}^{n\times p}$, $G\in \mathbb{R}^{m\times n}$, and
$H\in \mathbb{R}^{m\times p}$ are the dynamic, input, output and
feedthrough matrices of the compensator, respectively.
The optimal LQG controller can be
conveniently obtained using the separation principle by concatenating
the Kalman filter for \eqref{eq: system in x} with the (static)
controller that solves the Linear Quadratic Regulator problem for
\eqref{eq: system in x} with weight matrices $Q_x$ and
$R_u$. Specifically, after some manipulation, the optimal input that
solves the LQG problem reads as \eqref{eq: compensator}, we refer the reader to Appendix \ref{app: optimal LQG controller} for the details.\\
In what follows, we will make use of an equivalent representation of
the system \eqref{eq: system in x}. To this aim, let
\begin{align}\label{eq: z}
  z(t)\triangleq[U(t-1)^{\T},Y(t)^{\T},
  W(t-1)^{\T}, V(t)^{\T}]^{\T} ,
\end{align}
where
\begin{align}
  U(t-1)&\triangleq\left[u(t-n)^{\T}, \cdots, u(t-1)^{\T}\right]^{\T}, \nonumber\\
  Y(t)&\triangleq\left[y(t-n)^{\T}, \cdots, y(t)^{\T}\right]^{\T}, \nonumber\\
  W(t-1)&\triangleq\left[w(t-n)^{\T}, \cdots, w(t-1)^{\T}\right]^{\T}, \nonumber\\
  V(t)&\triangleq\left[v(t-n)^{\T}, \cdots, v(t)^{\T}\right]^{\T}. \nonumber
\end{align}
We can write an equivalent representation of \eqref{eq: system in x} in the behavioral space $z$ as \eqref{eq: z dynamics} (see Appendix \ref{app: behavioral dynamics} for the derivation).
%
In fact, given a sequence of control inputs and noise values, the state $z$ contains the system output $y$ over time, and can be used to reconstruct the exact value of the system state~$x$. This also implies that a controller for the system \eqref{eq: system in x} can
equivalently be designed using the dynamics \eqref{eq: z dynamics}. In
fact, we show that any \emph{dynamic} controller for
\eqref{eq: system in x} can be equivalently represented as a
\emph{static} controller for \eqref{eq: z dynamics}, see Appendix \ref{app: dynamic to static controller}. Next, we reformulate the LQG problem \eqref{eq: LQG cost} for the behavioral dynamics \eqref{eq: z dynamics} and characterize its optimal solution. The LQG problem \eqref{eq: LQG cost} can be equivalently written in the behavioral space as:
\setcounter{equation}{5}
  \begin{align}\label{eq: LQG cost z}
    \mc{J}_z\triangleq\lim_{T\rightarrow \infty}\mathbb{E} \left[
    \frac{1}{T}\Big(\sum_{t=n}^{T-1}z(t)^{\T}Q_z z(t) +
    u(t)^{\T} R_u u(t)\Big) \right], 
  \end{align}
subject to \eqref{eq: z dynamics}, where $Q_z$ is presented in Appendix \ref{app: LQG problem in the behavioral space} along with the derivation of \eqref{eq: LQG cost z}, and $R_u$ is as in \eqref{eq: LQG cost}.
The solution to the LQG problem in the behavioral space is given by a static controller, which we characterize next.
\begin{theorem}{\bf \emph{(Behavioral solution to the LQG
      problem)}}\label{thrm: solution of lqg in z}
  Let $u^*$ be the minimizer of \eqref{eq: LQG cost z} subject to
  \eqref{eq: z dynamics}. Then,
  \begin{align}\label{eq: behavioral LQG gain}
    u^* (t) =& \underbrace{-\left(R_u+\mc{B}_u^{\T}M
    \mc{B}_u\right)^{-1} \mc{B}_u^{\T}M \mc{A} P
    \mc{C}^{\T}\left( \mc{C} P
    \mc{C}^{\T}\right)^{\dagger}}_{\mc K} y_z (t)
    \end{align}
  where $M\succeq 0$ and $P \succeq 0$ are the unique solutions of the
  following coupled Riccati equations:
  \begin{align*}
    M &= \mc{A}^{\T}M \mc{A} -\mc{A}^{\T}M\mc{B}_u S_M
        \mc{B}_u^{\T}M \mc{A} + Q_z\\
      &+\left(I-P\mc{C}^{\T}S_P\mc{C}\right)^{\T}\mc{A}^{\T}M\mc{B}_u
        S_M\mc{B}_u^{\T} M
        \mc{A}\left(I-P\mc{C}^{\T}S_P\mc{C} \right)\\
    P &=\mc{A}P \mc{A}^{\T} - \mc{A}P\mc{C}^{\T}S_P
        \mc{C}P \mc{A}^{\T}+\mc{B}_w Q_w\mc{B}_w^{\T}
        + \mc{B}_v R_v\mc{B}_v^{\T}\\
    +&\left(I-M\mc{B}_u S_M
       \mc{B}_u^{\T}\right)^{\T}\mc{A}P\mc{C}^{\T}S_P\mc{C}P
       \mc{A}^{\T} \left(I-M\mc{B}_uS_M\mc{B}_u^{\T}\right)
  \end{align*}
  with $S_M\triangleq(R_u+\mc{B}_u^{\T}M\mc{B}_u)^{-1}$ and
  $S_P\triangleq(\mc{C}P\mc{C}^{\T})^{\dagger}$.\oprocend
\end{theorem}
\smallskip
The proof of Theorem \ref{thrm: solution of lqg in z} is postponed to Appendix \ref{app: proof of theorem}. The gain $\mc{K}$ is not unique since $\mc{C}P\mc{C}^{\T}$ is generally not invertible. In some cases, such as with SISO systems, the gain $\mc{K}$ becomes unique, which gives solving for the optimal LQG controller in the behavioral space an advantage over solving for it in the state space. The issue of non-uniqueness of $\mc{K}$ stems from the fact that $y_z$ has components that are dependent on each other, which makes the left kernel of $\mc{C}P\mc{C}^{\T}$ non-empty. We can avoid this issue by carefully choosing the time window of $U$ and $Y$ that form the behavioral space in \eqref{eq: z}, but we leave this aspect for future work. Note that, solving the coupled Riccati equations that characterize the LQG gain in Theorem \ref{thrm: solution of lqg in z} can be challenging. One method to solve for the LQG gain is to solve for the LQR and the Kalman gains, then use \eqref{eq: state space LQG} and Lemma \ref{lemma: static controller in z}.

%
%
%
\begin{example}{\bf \emph{(LQG controller in the behavioral space)}} \label{ex: example_1}
Consider the system \eqref{eq: system in x} with $A=1.1$, $B=1$, $C=1$, $Q_w=0.5$, and $R_v=0.8$. Also, consider the optimal control problem \eqref{eq: LQG cost} with $Q_x=R_u=1$. The Kalman and the LQR gains are $\subscr{K}{kf}=0.5474$ and $\subscr{K}{lqr}=0.7034$, respectively, which can be written as \eqref{eq: compensator} using \eqref{eq: state space LQG} with $E=0.1716$, $F=0.0973$, $G=-0.7034$, and $H=-0.3991$. Using \eqref{eq: z}, we define the behavioral space as $z(t)\triangleq\left[u(t-1), y(t-1), y(t),w(t-1),v(t-1),v(t)\right]^{\T}$ for $t\geq 1$. Using Lemma \ref{lemma: system in z}, we write the equivalent dynamics of \eqref{eq: system in x} in the behavioral space as \eqref{eq: z dynamics} with $\mc{A}_u=0.4977$, $\mc{A}_y=\begin{bmatrix}0.5475 & 0.6023 \end{bmatrix}$, $\mc{A}_w=0.4977$, and $\mc{A}_y=\begin{bmatrix}-0.5475 & -0.6023 \end{bmatrix}$. Using Theorem \ref{thrm: solution of lqg in z}, the LQG gain is $\mc{K}=[0.1716,0,-0.3991]$. Fig. \ref{fig: example_1}(a) shows the free response of \eqref{eq: system in x} and \eqref{eq: z dynamics} with equal initial conditions. Fig. \ref{fig: example_1}(b) shows the response of \eqref{eq: system in x} and \eqref{eq: z dynamics} to the LQG controllers \eqref{eq: state space LQG} and \eqref{eq: behavioral LQG gain}, respectively.\oprocend
\end{example}
\begin{figure}[!t]
  \centering
  \includegraphics[width=1\columnwidth,trim={0cm 0cm 0cm
    0cm},clip]{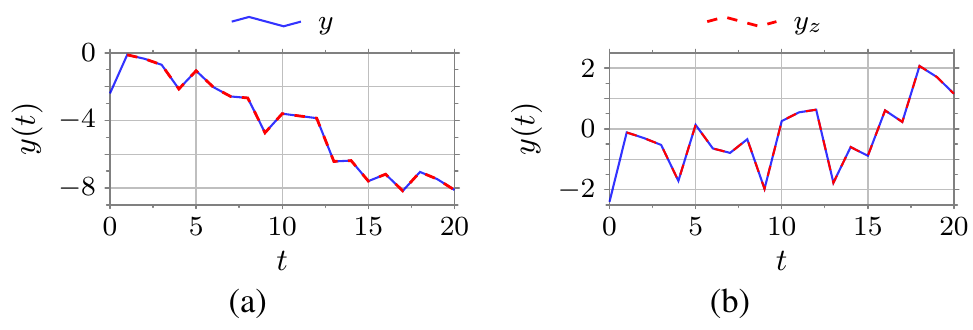}
  \caption{This figure shows the free response and the LQG feedback response of \eqref{eq: system in x} and \eqref{eq: z dynamics} for the setting defined in Example \ref{ex: example_1}. In both panels, the solid blue line and the dashed red line represent the output of \eqref{eq: system in x} and the output of \eqref{eq: z dynamics} that corresponds to $y(t)$, respectively. Panel (a) shows the free response of \eqref{eq: system in x} and \eqref{eq: z dynamics}, we observe that the response of both systems are equal, which agrees with Lemma \ref{lemma: system in z}. Panel (b) shows the feedback response of \eqref{eq: system in x} and \eqref{eq: z dynamics} to the LQG controller \eqref{eq: state space LQG} and the behavioral LQG controller in Theorem \ref{thrm: solution of lqg in z}, respectively. We observe that both systems have equal responses, which agrees with Lemma \ref{lemma: lqg in z} and Theorem \ref{thrm: solution of lqg in z}. Notice that the response of \eqref{eq: z dynamics} starts at time $t=n=1$ since we have to wait $n=1$ time steps in order to get the equivalent initial condition for \eqref{eq: z dynamics}.}
    \label{fig: example_1}
\end{figure}

\section{Implications of behavioral representation in numerical methods} \label{sec: numerical methods}
In this section, we highlight some implications of our behavioral representation and results. In particular, we provide an analysis of learning the LQG controller from finite expert demonstrations, and an analysis of solving for the behavioral LQG gain via a gradient descent method. First, we present the following Lemma regarding the sparsity of the LQG gain in \eqref{eq: behavioral LQG gain}, which we use in our subsequent analysis.
\begin{lemma}{\bf \emph{(Sparsity of the optimal LQG gain)}}\label{lemma: sparsity of the LQG gain} 
Consider the LQG gain written in the behavioral space as
\begin{align}\label{eq: lqg controller rewritten}
u(t)=\left[ \begin{array}{ccc}
\mc{K}_1 & \mc{K}_{2} & \mc{K}_{3}
\end{array}\right]
\left[\begin{array}{c}
U(t-1)\\
y(t-n)\\
\overline{Y}(t)
\end{array}\right],
\end{align}
where $\overline{Y}(t)\triangleq \left[y(t-n+1)^{\T}, \cdots, y(t)^{\T}\right]^{\T}$. Then $\mc{K}_{2}=0$.~\oprocend
\end{lemma}
\smallskip
A proof of Lemma \ref{lemma: sparsity of the LQG gain} is presented in Appendix \ref{app: sparsity of the LQG gain}.
\subsection{Learning LQG controller from expert demonstrations}
Consider the system \eqref{eq: system in x}, assume that the system is stabilized by an expert that uses optimal LQG controller. We also assume that the system and the noise statistics are unknown. Our objective is to learn the optimal LQG controller from finite expert demonstrations, which are composed of input and output data. In the behavioral representation, this boils down to learning the gain $\mc{K}$ of the subspace $u(t)=\mc{K}y_z(t)$ for $t\geq n$. Using Lemma \ref{lemma: sparsity of the LQG gain}, we only need to learn $\mc{K}_1$ and $\mc{K}_3$, which are obtained as $\left[ \mc{K}_1 \ \mc{K}_3 \right]=U_N Y_N^{\dagger} + \mc{K}_{\text{null}}$, where
\begin{align}\label{eq: expert data}
\begin{split}
U_N&\triangleq
\left[\begin{smallarray}{ccc}
u(t) & \cdots & u(t+k-1)
\end{smallarray}\right],\\
Y_N& \triangleq
\left[\begin{smallarray}{ccc}
u(t-n)     & \cdots   & u(t-n+k-1)\\
\vdots    & \ddots   & \vdots \\
u(t-1)     &  \cdots  &  u(t-2+k)\\
y(t-n+1) &  \cdots  & y(t-n+k) \\
\vdots    & \ddots   & \vdots \\
y(t)        & \cdots   & y(t-1+k)
\end{smallarray}\right],
\end{split}
\end{align}
for $t\geq n$, where $k$ is the number of columns, and $\mc{K}_{\text{null}}$ is any matrix with appropriate dimension whose rows belong to the left null space of $Y_N$. Note that $\mc{K}_{\text{null}}$ will disappear when multiplied by the feedback $y_z(t)$, i.e., $\mc{K}_{\text{null}}y_z(t)=0 $. Therefore, without loss of generality, we set $\mc{K}_{\text{null}}=0$.


\begin{lemma}{\bf \emph{(Sufficient number of expert data to compute the optimal LQG gain)}}\label{lemma: number of expert data} 
Consider input and output expert samples $U=[u(t),\cdots,u(t+N-1)]$ and $Y=[y(t),\cdots,y(t+N-1)]$ 
generated by LQG controller to stabilize system \eqref{eq: system in x}, such that $U$ is full-rank. 
Then, $N=n+nm+np$  expert samples are sufficient to compute the LQG gain $\mc{K}$.\oprocend
\end{lemma}
\smallskip
A proof of Lemma \ref{lemma: number of expert data} is presented in Appendix \ref{app: number of expert data}. We note that the rank condition on the input matrix~$U$ in the statement of Lemma~\ref{lemma: number of expert data}
is a reasonable assumption owing to the fact that system~\eqref{eq: system in x} is driven by i.i.d. process noise~$w$ and that 
the measurement noise~$v$ is also i.i.d.
Furthermore, note that we can learn the dynamic controller matrices $E$, $F$, $G$, and $H$ in \eqref{eq: compensator} (up to a similarity transformation) using subspace identification methods for deterministic systems (see \cite{PVO-BDM:96}) with $U$ and $Y$ treated as the output and input signals to \eqref{eq: compensator}, respectively. Using \cite[Theorem 2]{PVO-BDM:96}, we need at least $N= 2(n+1)(m+p+1)-1$ expert samples to learn \eqref{eq: compensator}, which is more than the sufficient number of expert samples to learn $\mc{K}$ (Lemma \ref{lemma: number of expert data}). 
\begin{example}{\bf \emph{(Learning LQG controller from expert data)}} \label{ex: example_3}
Consider the system in Example \ref{ex: example_1} where the system dynamics and the noise statistics are assumed to be unknown. The system is driven by an expert that uses an LQG policy. According to Lemma \ref{lemma: number of expert data}, we collect $N=n+nm+np=3$ expert input-output samples to form the data matrices 
\begin{align*}
U_N=\begin{bmatrix}
-0.2269 & -0.1231
\end{bmatrix},\quad 
Y_N=\begin{bmatrix}
1.7878 & -0.2269\\
1.3371 & 0.211
\end{bmatrix}.
\end{align*}
Using the data, we obtain $[\mc{K}_1 \ \mc{K}_3]=[0.1716 \ -0.3991]$~with $\mc{K}_{\text{null}}=0$, which matches the LQG gain in Example \ref{ex: example_1}.\oprocend
\end{example}

\subsection{Gradient descent in the behavioral space}
In this section, we use gradient descent to solve for $\mc{K}$:
\begin{align}\label{eq: grad descent}
\mc{K}^{(i+1)}=\mc{K}^{(i)} - \alpha^{(i)} \nabla\mc{J}_z(\mc{K}^{(i)})
\quad \text{for } i=0,1,2, \cdots
\end{align}
where the index $i$ refers to the iteration number, $\alpha^{(i)}$ is the step size at iteration $i$, and $\nabla\mc{J}_z(\mc{K}^{(i)})$ is computed using \eqref{eq: derivative of Jz}. We initialize the gradient descent method with a stabilizing gain $\mc{K}^{(0)}$. We determine the step size $\alpha^{(i)}$ by the Armijo rule \cite[Chapter $1.3$]{DPB:95B}: initialize $\alpha^{(0)}=1$, repeat $\alpha^{(i)}=\beta \alpha^{(i)} $ until 
\begin{align*}
\mc{J}_z(\mc{K}^{(i+1)})\leq \mc{J}_z(\mc{K}^{(i)}) -\sigma\alpha^{(i)} \left\| \nabla\mc{J}_z(\mc{K}^{(i)})\right\|_{\text{F}}^2
\end{align*}
is satisfied, with $\beta,\sigma \in (0,1)$.
\begin{example}{\bf \emph{(Gradient descent)}} \label{ex: example_4}
We consider the example in \cite{JCD:78} discretized with sampling time $T_s=0.4$,
\begin{align*}
A&=\begin{bmatrix}
1.4918 & 0.5967 \\
0          & 1.4918
\end{bmatrix}, \ 
B=\begin{bmatrix}
0.1049\\
0.4918
\end{bmatrix},\ 
C=\begin{bmatrix}
1 & 0
\end{bmatrix},\\
Q_w&=\begin{bmatrix}
4.6477 & 3.7575 \\
3.7575 & 3.0639
\end{bmatrix},\quad
Q_x=\begin{bmatrix}
3.0639 & 3.7575 \\
3.7575 & 4.6477
\end{bmatrix},\quad
\end{align*}
$R_v=2.5$ and $R_u=0.5966$. The LQG gain from Theorem \ref{thrm: solution of lqg in z} is $\mc{K}=[-0.0366, -0.103,0,5.8461, -4.7434]$. Using Lemma \ref{lemma: sparsity of the LQG gain}, we only need to do the search over $\mc{K}_1$ and $\mc{K}_3$ since $\mc{K}_2=0$. We use gradient descent in \eqref{eq: grad descent} to solve for the LQG gain. We choose a stabilizing initial gain $\mc{K}^{(0)}$ that randomly place the closed-loop eigenvalues within $[0.45,0.92]$. We use the Armijo rule to compute the step size with $\alpha^{(0)}=1$, $\beta=0.8$, and $\sigma=0.7$. We set the stopping criteria to be when the gradient vanishes or when the maximum number of iterations is reached (in this example we set it to $15000$ iterations). For numerical comparison, we use gradient descent to solve for the optimal LQG dynamic controller in the form of \eqref{eq: compensator} as in \cite{YZ-YT-NL:21}, where we optimize the LQG cost \eqref{eq: LQG cost} and apply the gradient search over the control parameters $E$, $F$, $G$, and $H$.{\footnote{In \cite{YZ-YT-NL:21}, $H=0$ since it is assumed that the control input $u(t)$ at time $t$ depends on the history $\{u(0),\cdots,u(t-1),y(0),\cdots,y(t-1)\}$. In this paper, $u(t)$ depends also on $y(t)$, therefore $H$ is nonzero (see Appendix~\ref{app: optimal LQG controller}). We computed the gradient of $\mc{J}$ w.r.t. the controller matrices $E$, $F$, $G$ and $H$ as in \cite{YZ-YT-NL:21} adapted to the case where $H$ is nonzero. We have not included the derivations in this paper due to space constraint.}} Fig. \ref{fig: example_4} shows the convergence performance of the gradient descent for different initial conditions. We observe that the gradient descent over $\mc{K}$ in Fig. \ref{fig: example_4}(a) converges to $\mc{K}^*=[-0.0366, -0.1030,0,5.8460, -4.7434]$ before reaching the maximum number iterations for different initial conditions. Starting from initial conditions equivalent to the ones in Fig. \ref{fig: example_4}(a), the gradient descent over the controller matrices $E$, $F$, $G$ and $H$ in Fig. \ref{fig: example_4}(b) did not converge within~$15000$~iterations.~\oprocend
\end{example}
%
\begin{figure}[!t]
  \centering
  \includegraphics[width=1\columnwidth,trim={0cm 0cm 0cm
    0cm},clip]{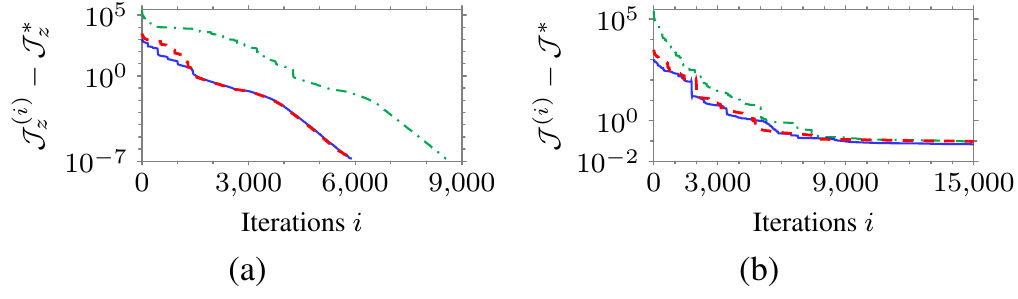}
  \caption{This figure shows the convergence performance (measured by the suboptimality gap) of the gradient descent applied to the system in Example \ref{ex: example_4}. The solid blue line, dashed red line and the dash-dotted green line represent different initial conditions, respectively. Panels (a) and (b) show the convergence performance of the gradient descent over $\mc{K}$ and the gradient descent over the controller matrices $E$, $F$, $G$ and $H$,~respectively.}
    \label{fig: example_4}
\end{figure}
\section{Conclusion and future work}
In this work, we revisited the LQG optimal control problem from a behavioral perspective. We introduced equivalent representations for the class of stochastic discrete-time, linear, time-invariant systems and the LQG optimal control problem in the space of input-output behaviors. In particular, we showed that the optimal LQG controller can be expressed as a static behavioral-feedback gain, which can be solved for directly from the LQG problem in the behavioral space. Finally, we highlighted the advantages of having a static LQG gain over a dynamic LQG controller in the context of data-driven control and gradient-based algorithms, which arise from the fact that the behavioral approach circumvents the need for a state space representation and the fact that the optimal behavioral-feedback is a static gain. 
There still remain several unexplored questions, including the investigation of the optimization landscape of the LQG problem in the behavioral space, which will pave the way for an improved understanding of the convergence properties of data-driven and gradient algorithms, as well as, for investigating the uniqueness of the optimal LQG gain.
%
%
%
%
\bibliographystyle{unsrt}
\bibliography{alias,Main,FP,New}
%
%
%
\appendix
\setcounter{section}{5} 
\subsection{Optimal LQG controller}\label{app: optimal LQG controller}
The optimal LQG controller that solves \eqref{eq: LQG cost} is written as
\begin{align}\label{eq: LQG compensator}
\begin{split}
\hat{x}(t+1) &= (I_n-\subscr{K}{kf}C)(A-B\subscr{K}{LQR})\hat{x}(t)+\subscr{K}{kf} y(t+1), \\
u(t) &= -\subscr{K}{LQR}\hat{x}(t),
\end{split}
\end{align}
where $\subscr{K}{kf}$ and $\subscr{K}{LQR}$ are the Kalman and
LQR gains, respectively. To write the controller \eqref{eq: LQG compensator} in the form of \eqref{eq: compensator}, we need the following lemma.
\begin{appxlem}{\bf \emph{(Equivalent compensator forms)}}\label{lemma: compensator forms} Consider the compensator \eqref{eq: compensator} and a compensator of the form:
\begin{align}\label{eq: compensator 2}
\begin{aligned}
\xi_c(t+1) &= \overbar{E}\xi_c(t)+\overbar{F}y(t+1), \\
u(t) &= \overbar{G}\xi_c(t),
\end{aligned}
\end{align}
with $\xi_c \in \mathbb{R}^n$ denoting the state, and $\overbar{E}\in \mathbb{R}^{n\times n}$,
$\overbar{F}\in \mathbb{R}^{n\times p}$, and $\overbar{G}\in \mathbb{R}^{m\times n}$ denoting the dynamic, input, and output matrices of the compensator, respectively. Let $x_c(0)=\xi_c(0)$ and $y(0)=0$, then, the compensators \eqref{eq: compensator} and \eqref{eq: compensator 2} output the same $u(t)$ given the same input $y(t)$ if:
\begin{align}\label{eq: compensator matrices}
\begin{split}
E=\overbar{E}, \quad F=\overbar{E} \overbar{F}, \quad G=\overbar{G}, \quad H=\overbar{G}\overbar{F}.
\end{split}
\end{align}
\end{appxlem}
\begin{proof}
Using \eqref{eq: compensator} with $y(0)=0$, we can write
\begin{align}\label{eq: input evolution 1}
u(t)=G E^t x_c(0) + 
\begin{bmatrix}
G E^{t-2} F & \cdots &GF  & H
\end{bmatrix}
y,
\end{align}
where $y=[y(1)^{\transpose},\cdots, y(t)^{\transpose}]^{\transpose}$. Using \eqref{eq: compensator 2}, we can write
\begin{align}\label{eq: input evolution 2}
u(t)=\overbar{G} \overbar{E}^t \xi_c(0) + 
\begin{bmatrix}
\overbar{G} \overbar{E}^{t-1} \overbar{F} & \cdots & \overbar{G}\overbar{F}
\end{bmatrix}
y.
\end{align}
Under the same $y$, \eqref{eq: input evolution 1} is equal to \eqref{eq: input evolution 2} for $E=\overbar{E}$, $F=\overbar{E} \overbar{F}$, $G=\overbar{G}$, and $H=\overbar{G}\overbar{F}$.
\end{proof}
Using Lemma \ref{lemma: compensator forms} and \eqref{eq: LQG compensator}, we can write the LQG controller in the form of \eqref{eq: compensator} with
\begin{align}\label{eq: state space LQG}
\begin{split}
E&=(I_n-\subscr{K}{kf}C)(A-B\subscr{K}{LQR}),\\
F&=(I_n-\subscr{K}{kf}C)(A-B\subscr{K}{LQR})\subscr{K}{kf},\\
G&=-\subscr{K}{LQR},\\
H&=-\subscr{K}{LQR}\subscr{K}{kf}.
\end{split}
\end{align}
\subsection{System representation in the behavioral space}\label{app: behavioral dynamics}
The following Lemma provides an equivalent representation of
\eqref{eq: system in x} in the behavioral space, $z$, which is written in~\eqref{eq: z dynamics}.

\begin{appxlem}{\bf \emph{(Equivalent dynamics)}}\label{lemma:
    system in z}
  Let $z$ be as in~\eqref{eq: z}. Then,
  \begin{align*}
    z(t+1) = \mc{A} z(t) + \mc B_u u(t) + \mc B_w w(t) + \mc B_v
    v(t+1),
  \end{align*}
  where $\mc A$, $\mc B_u$, $\mc B_w$, and $\mc B_v$ are as in
  \eqref{eq: z dynamics}, and 

  \begin{align*}
\mc{A}_u&\triangleq\mc{F}_2-CA^{n+1}\mc{O}^{\dagger} \mc{F}_1, \qquad \mc{A}_y\triangleq CA^{n+1}\mc{O}^{\dagger},\\
\mc{A}_w&\triangleq\mc{F}_4-CA^{n+1}\mc{O}^{\dagger}\mc{F}_3, \qquad \mc{A}_v \triangleq -CA^{n+1}\mc{O}^{\dagger},
  \end{align*}
  \begin{align*}
    \mc{O}&\triangleq \left[\begin{smallarray}{c}
      C\\
      CA\\
      \vdots \\
      CA^n
    \end{smallarray}\right], \quad 
    \mc{F}_1\triangleq\left[\begin{smallarray}{ccc}
      0    & \cdots & 0\\
      CB & \cdots & 0\\
      \vdots  & \ddots & \vdots \\
      CA^{n-1}B & \cdots & CB
    \end{smallarray}\right],\\
    \mc{F}_2&\triangleq\left[\begin{smallarray}{ccc}
      CA^nB & \cdots & CAB
    \end{smallarray}\right],
  \end{align*}
  and the matrices $\mc{F}_3$ and $\mc{F}_4$ are obtained by replacing
  $B$ with $I$ in $\mc{F}_1$ and $\mc{F}_2$, respectively.
\end{appxlem}
\begin{proof}
We can write the evolution of $y(t+1)$ as
\begin{align} \label{eq: evolution of y}
y(t+1)=&CA^{n+1}x(t-n) +\mc{F}_2 U(t-1) + \mc{F}_4 W(t-1)\nonumber \\
 &+ CBu(t) +Cw(t) +v(t+1),
\end{align}
where, $\mc{F}_2$ and $\mc{F}_4$ are as in Lemma \ref{lemma: system in z}, and $U(t-1)$ and $W(t-1)$ are as in \eqref{eq: z}. Also, we can write $Y(t)$ in \eqref{eq: z} in terms of $U(t-1)$, $W(t-1)$, and $V(t)$:
\begin{align}\label{eq: evolution of Y}
Y(t)=\mc{O} x(t-n) + \mc{F}_1 U(t-1) + \mc{F}_3 W(t-1) + V(t),
\end{align}
where $\mc{O}$, $\mc{F}_1$, and $\mc{F}_3$ are same as in Lemma \ref{lemma: system in z} and $V(t)$ is as in \eqref{eq: z}. Then using \eqref{eq: evolution of Y}, we substitute $x(t-n)$ into \eqref{eq: evolution of y} to get
\begin{align*}
y(t+1)=&\begin{bmatrix}
\mc{A}_u & \mc{A}_y & \mc{A}_w & \mc{A}_v
\end{bmatrix}\underbrace{\begin{bmatrix}
U(t-1)\\
Y(t)\\
W(t-1)\\
V(t)
\end{bmatrix}}_{z(t)}+CB u(t)\\
& +C w(t) +v(t+1),
\end{align*}
where $\mc{A}_u$, $\mc{A}_y$, $\mc{A}_w$, and $\mc{A}_v$ are as in Lemma \ref{lemma: system in z}.
\end{proof}

\subsection{From dynamic to static controller}\label{app: dynamic to static controller}

\begin{appxlem}{\bf \emph{(From dynamic to static
      controllers)}}\label{lemma: static controller in z}
  Let the control input $u$ be the output of the dynamic controller
  \eqref{eq: compensator}. Then, equivalently,
  \begin{align}\label{eq: static controller in z}
    u(t)=
    \begin{bmatrix}
        GE^{n}\mc{T}_1^{\dagger} & \mc{T}_2 - GE^{n}\mc{T}_1^{\dagger} \mc{M}
      \end{bmatrix}
                                   y_z(t),
  \end{align}
  where $y_z$ is as in \eqref{eq: z dynamics}, and
  \begin{align*}
    \mc{T}_1&\triangleq\left[\begin{smallarray}{c}
      G\\
      GE\\
      \vdots\\
      GE^{n-1}
    \end{smallarray}\right], \quad
    \mc{T}_2\triangleq\left[\begin{smallarray}{cccc}
      GE^{n-1}F & \cdots & GF & H
    \end{smallarray}\right],\\
    \mc{M}&\triangleq\left[\begin{smallarray}{cccccc}
      H & 0 & 0 &\cdots & 0 & 0\\
      \vspace{0.2cm}
      GF & H & 0 & \cdots & 0 & 0 \\
      GEF & GF & H & \cdots & 0 & 0\\
      \vspace{0.2cm}
      \vdots & \vdots  & \vdots &\ddots & \vdots & \vdots\\
      GE^{n-2}F& GE^{n-3}F & \cdots& \cdots &H & 0
    \end{smallarray}\right].
  \end{align*}
\end{appxlem}
\smallskip
\begin{proof}
Using \eqref{eq: compensator}, we can write
\begin{align}\label{eq: evolution of u}
u(t)=GE^n x_c(t-n) + \mc{T}_2 Y(t),
\end{align}
where $\mc{T}_2$ and $Y(t)$ are as in Lemma \ref{lemma: static controller in z} and \eqref{eq: z}, respectively. Further, we can write $U(t-1)$ in \eqref{eq: z} as
\begin{align}\label{eq: evolution of U}
U(t-1)=\mc{T}_1 x_c(t-n) + \mc{M} Y(t),
\end{align} 
where $\mc{T}_1$ and $\mc{M}$ are as in Lemma \ref{lemma: static controller in z}. Using \eqref{eq: evolution of U} we substitute $x_c(t-n)$ into \eqref{eq: evolution of u} to get
\begin{align*}
u(t)=\begin{bmatrix}
GE^n \mc{T}_1^{\dagger} & \mc{T}_2-GE^n\mc{T}_1^{\dagger} \mc{M}
\end{bmatrix} \underbrace{\begin{bmatrix}
U(t-1)\\
Y(t)
\end{bmatrix}}_{y_z}.
\end{align*}
\end{proof}
\subsection{LQG problem in the behavioral space}\label{app: LQG problem in the behavioral space}

\begin{appxlem}{\bf \emph{(LQG problem in the behavioral
      space)}}\label{lemma: lqg in z}
  The input $u^*$ is the minimizer of \eqref{eq: LQG cost} subject to
  \eqref{eq: system in x} if and only if it is the minimizer of
  \begin{align}\label{eq: LQG cost z app}
    \mc{J}_z\triangleq\lim_{T\rightarrow \infty}\mathbb{E} \left[
    \frac{1}{T}\Big(\sum_{t=0}^{T-1}z(t)^{\T}Q_z z(t) +
    u(t)^{\T} R_u u(t)\Big) \right] 
  \end{align}
  subject to \eqref{eq: z dynamics},
  where $Q_z=\mc{H}^{\T}Q_x\mc{H}$ and
  \begin{align*}
    \mc{H}&\triangleq\begin{bmatrix}
      \mc{G}_1-A^n \mc{O}^{\dagger}\mc{F}_1 & A^n \mc{O}^{\dagger} & \mc{G}_2-A^n \mc{O}^{\dagger}\mc{F}_3 & -A^n \mc{O}^{\dagger}
    \end{bmatrix},\\
    \mc{G}_1&\triangleq\begin{bmatrix}
      A^{n-1}B & \cdots & B
    \end{bmatrix}, \quad
                          \mc{G}_2\triangleq\begin{bmatrix}
                            A^{n-1} & \cdots & I_n
                          \end{bmatrix}.
  \end{align*}
\end{appxlem}
\smallskip
\begin{proof}
We begin by proving that the costs in \eqref{eq: LQG cost} and \eqref{eq: LQG cost z app} are equivalent. We can express $x(t)$ for $t\geq n$ as
\begin{align}\label{eq: evolution of x}
x(t)=A^n x(t-n) +\mc{G}_1 U(t-1) + \mc{G}_2 W(t-1),
\end{align}
where $\mc{G}_1$ and $\mc{G}_2$ are as in Lemma \ref{lemma: lqg in z}, and $U(t-1)$ and $W(t-1)$ are as in \eqref{eq: z}. Using \eqref{eq: evolution of Y}, we can substitute $\left. x(t-n)\right.$ in terms of $U(t-1)$, $Y(t)$, $W(t-1)$, and $V(t)$ into \eqref{eq: evolution of x} to get $x(t)=\mc{H} z(t)$, where $\mc{H}$ is as in Lemma \ref{lemma: lqg in z}. Substituting $x(t)=\mc{H} z(t)$ into the cost \eqref{eq: LQG cost} yields the cost \eqref{eq: LQG cost z  app}. Further, Lemma \ref{lemma: system in z} shows that the systems \eqref{eq: system in x} and \eqref{eq: z dynamics} are equivalent. Therefore, the minimizer of \eqref{eq: LQG cost} subject to \eqref{eq: system in x} is the minimizer of \eqref{eq: LQG cost z  app} subject to \eqref{eq: z dynamics}.
\end{proof}

\subsection{Proof of Theorem \ref{thrm: solution of lqg in z}} \label{app: proof of theorem}

For the proof of Theorem \ref{thrm: solution of lqg in z}, we need the following technical results from the literature. 
\begin{appxlem}{\bf \emph{(Steady-state cost)}}\label{lemma: steady-state cost}
For a controller $u(t)=\mc{K}y_z(t)$ with stabilizing gain $\mathcal{K}$, the cost \eqref{eq: LQG cost z} at steady-state is written as
 \begin{align}\label{eq: LQG cost z steady-state}
 \mathcal{J}_z(\mathcal{K})=\Tr{Q_{\mathcal{K}}P},
 \end{align}
 where $Q_{\mathcal{K}}\triangleq Q_z+\mathcal{C}^{\T}\mathcal{K}^{\T}R_u\mathcal{K}\mathcal{C}$, and $P\succeq 0$ is the unique solution of the following Lyapunov equation
 \begin{align}\label{eq: P lyap}
 P= \mathcal{A}_c P\mathcal{A}_c^{\T}+\mathcal{B}_w Q_w \mathcal{B}_w^{\T}+\mathcal{B}_v R_v \mathcal{B}_v^{\T}.
 \end{align}
 with $\mathcal{A}_c \triangleq \mathcal{A}+\mathcal{B}_u\mathcal{K}\mathcal{C}$.
\end{appxlem}
\smallskip
\begin{proof}
Since $u(t)=\mc{K}y_z(t)$ is stabilizing, the closed-loop matrix $\mc{A}_c=\mc{A}+\mc{B}_u \mc{K}\mc{C}$ is stable. We can write
\begin{align*}
\expect{z(t)z(t)^{\T}}=&\mc{A}_c\expect{z(t-1)z(t-1)^{\T}}\mc{A}_c^{\T}+\mc{B}_w Q_w \mc{B}_w^{\T}\\
 &+\mc{B}_v R_v \mc{B}_v^{\T},
\end{align*}
where we have used the fact that $z(t-1)$, $w(t-1)$, and $v(t)$ are uncorrelated, and $\expect{w(t-1)w(t-1)^{\T}}=Q_w$ and $\expect{v(t)v(t)^{\T}}=R_v$. Since $\mc{A}_c$ is stable, $\underset{t\rightarrow \infty}{\lim} \expect{z(t)z(t)^{\T}}$ converges to a finite value, and at steady state we have $P\triangleq \underset{t\rightarrow \infty}{\lim} \expect{z(t)z(t)^{\T}}=\underset{t\rightarrow \infty}{\lim} \expect{z(t-1)z(t-1)^{\T}}$, where $P$ satisfies \eqref{eq: P lyap}. The cost  \eqref{eq: LQG cost z} is written as
\begin{align*}
\mc{J}_z& \triangleq \lim_{T\rightarrow \infty}\mathbb{E} \left[\frac{1}{T}\left(\sum_{t=0}^{T-1}z(t)^{\T}\left(Q_z+\mc{C}^{\T}\mc{K}^{\T}R_u\mc{K}\mc{C}\right)z(t)\right) \right]\\
=&\lim_{t\rightarrow \infty}\expect{\Tr{z(t)^{\T}Q_{\mc{K}}z(t)} }= \Tr{Q_{\mc{K}}\lim_{t\rightarrow \infty}\expect{z(t)z(t)^{\T}}}\\
=&\Tr{Q_{\mc{K}}P},
\end{align*}
where $Q_{\mc{K}}\triangleq Q_z+\mc{C}^{\T}\mc{K}^{\T}R_u\mc{K}\mc{C}$. The proof is complete.
\end{proof}
\smallskip
\begin{appxlem}{\bf \emph{(Property of the solution to Lyapunov
      equation, \cite{AAALM-VK-FP:19b})}}\label{lemma: tech_res}
  Let $A$, $B$, $Q$ be matrices of appropriate dimensions with
  $\rho(A)<1$. Let $Y$ satisfy $Y = AYA^{\T} + Q$. Then,
  $\Tr{BY} = \Tr{Q^{\T}M}$, where $M$ satisfies
  $M = A^{\T}MA + B^{\T}$.\oprocend
\end{appxlem}
\begin{pfof}{Theorem \ref{thrm: solution of lqg in z}}
Using Lemma \ref{lemma: steady-state cost}, we can write the cost \eqref{eq: LQG cost z} at steady-state as \eqref{eq: LQG cost z steady-state}. Next, we compute the derivative of $\mathcal{J}_z(\mathcal{K})$ with respect to the variable $\mathcal{K}$. Taking the differential of \eqref{eq: P lyap} with respect to the variable $\mathcal{K}$, we get
\begin{align*}
&dP=\mathcal{A}_c dP\mathcal{A}_c^{\T} + d\mathcal{A}_c P\mathcal{A}_c^{\T} + \mathcal{A}_c P d\mathcal{A}_c^{\T} \triangleq \mathcal{A}_c dP\mathcal{A}_c^{\T} +X\\
&\implies \Tr{Q_{\mathcal{K}}dP} \overset{\text{(a)}}{=} \Tr{X M}\overset{\text{(b)}}{=}2\Tr{\mc{C}P\mc{A}_c^{\T}M\mc{B}_u d\mc{K}},
\end{align*}
where $M\succeq 0$ satisfies $M=\mc{A}_c^{\T}M \mc{A}_c + Q_{\mc{K}}$, (a) follows from Lemma \ref{lemma: tech_res}, and (b) follows from $\Tr{d\mathcal{A}_c P\mathcal{A}_c^{\T}M}=\Tr{(d\mathcal{A}_c P\mathcal{A}_c^{\T}M)^{\T}}$ and using the trace cyclic property. Taking the differential of $Q_{\mc{K}}$, we get 
\begin{align*}
&d Q_{\mc{K}} = \mathcal{C}^{\T}d \mathcal{K}^{\T}R_u\mathcal{K}\mathcal{C} +\mathcal{C}^{\T}\mathcal{K}^{\T}R_u d \mathcal{K}\mathcal{C}\\
&\implies \Tr{d Q_{\mc{K}}P}\overset{\text{(c)}}{=}2\Tr{\mc{C}P\mc{C}^{\T}\mc{K}^{\T}R_u d\mc{K}},
\end{align*}
where (c) follows similarly as (b). For notational convenience, we denote $\mc{J}_z(\mc{K})$ by $\mc{J}_z$. Taking the differential of $\mc{J}_z$ in \eqref{eq: LQG cost z steady-state}, we get,  
\begin{align} \label{eq: derivative of Jz}
&d\mc{J}_z=d\Tr{Q_{\mc{K}}P}=\Tr{dQ_{\mc{K}}P}+\Tr{Q_{\mc{K}}dP} \nonumber\\
&\quad~~=2\Tr{\left(\mc{C}P\mc{C}^{\T}\mc{K}^{\T}R_u+\mc{C}P\mc{A}_c^{\T}M\mc{B}_u \right)d\mc{K}} \nonumber \\
&\implies \frac{d\mc{J}_z}{d\mc{K}}=2\left(R_u\mc{K}\mc{C}P\mc{C}^{\T}+\mc{B}_u^{\T}M\mc{A}_c P\mc{C}^{\T}\right)\\
&\qquad\qquad\ \; =2\left(R_u + \mc{B}_u^{\T}M\mc{B}_u\right)\mc{K}\mc{C}P\mc{C}^{\T}+2\mc{B}_u^{\T}M\mc{A}P\mc{C}^{\T} \nonumber
\end{align}
The stationary optimality condition implies $\frac{d\mc{J}_z}{d\mc{K}}=0$, we get
\begin{align}\label{eq: behavioral LQG gain proof}
\mc{K}=-\left(R_u + \mc{B}_u^{\T}M\mc{B}_u\right)^{-1} \mc{B}^{\T}M\mc{A}P\mc{C}^{\T}\left(\mc{C}P\mc{C}^{\T}\right)^{\dagger}+\mc{K}_{\text{null}},
\end{align}
where we have used the right pseudo inverse of $\mc{C}P\mc{C}^{\T}$ since it is rank deficient, and $\mc{K}_{\text{null}}$ is any matrix with appropriate dimension whose rows belong to the left null space of $\mc{C}P\mc{C}^{\T}$. Next we derive the Riccati equations of $M$ and $P$. Let $S_M\triangleq(R_u+\mc{B}_u^{\T}M\mc{B}_u)^{-1}$ and
  $S_P\triangleq(\mc{C}P\mc{C}^{\T})^{\dagger}$. Substituting the expression of $\mc{K}$ in \eqref{eq: behavioral LQG gain proof} into \eqref{eq: P lyap}, we get
\begin{align*}
P&=\mc{A}P\mc{A}^{\T} -\mc{A}P\mc{C}^{\T}S_p\mc{C}P\mc{A}^{\T}M \mc{B}_u S_M\mc{B}_u^{\T}\\
&-\mc{B}_u S_M\mc{B}_u^{\T}M\mc{A}P\mc{C}^{\T} S_P \mc{C}P\mc{A}^{\T} +\mc{B}_w Q_w \mc{B}_w^{\T} +\mc{B}_v R_v \mc{B}_v^{\T}\\
&+\mc{B}_u S_M\mc{B}_u^{\T}M\mc{A}P\mc{C}^{\T}\underbrace{S_P \left(\mc{C}P\mc{C}^{\T}\right) S_P}_{\overset{\text{(d)}}{=}S_p} \mc{C}P\mc{A}^{\T}M \mc{B}_u S_M \mc{B}_u^{\T}\\
\overset{\text{(e)}}{=}&\mc{A}P\mc{A}^{\T} -\mc{A}P\mc{C}^{\T}S_p\mc{C}P\mc{A}^{\T}M \mc{B}_u S_M\mc{B}_u^{\T}\\
&~\;-\mc{B}_u S_M\mc{B}_u^{\T}M\mc{A}P\mc{C}^{\T} S_P \mc{C}P\mc{A}^{\T} +\mc{B}_w Q_w \mc{B}_w^{\T} +\mc{B}_v R_v \mc{B}_v^{\T}\\
&~\;+\mc{B}_u S_M\mc{B}_u^{\T}M\mc{A}P\mc{C}^{\T}S_P \mc{C}P\mc{A}^{\T}M \mc{B}_u S_M \mc{B}_u^{\T}\\
&~\; +\mc{A}P\mc{C}^{\T}S_p\mc{C}P\mc{A}^{\T} - \mc{A}P\mc{C}^{\T}S_p\mc{C}P\mc{A}^{\T}\\
=&\mc{A}P \mc{A}^{\T} - \mc{A}P\mc{C}^{\T}S_P
        \mc{C}P \mc{A}^{\T}+\mc{B}_w Q_w\mc{B}_w^{\T}
        + \mc{B}_v R_v\mc{B}_v^{\T}\\
    +&\left(I-M\mc{B}_u S_M
       \mc{B}_u^{\T}\right)^{\T}\mc{A}P\mc{C}^{\T}S_P\mc{C}P
       \mc{A}^{\T} \left(I-M\mc{B}_uS_M\mc{B}_u^{\T}\right),
\end{align*}
where (d) follows from the Moore-Penrose conditions, and in (e) we have added and subtracted the term $\mc{A}P\mc{C}^{\T}S_p\mc{C}P\mc{A}^{\T}$. The Riccati equation of $M$ is derived in similar manner.
\end{pfof}

\subsection{Proof of Lemma \ref{lemma: sparsity of the LQG gain}}\label{app: sparsity of the LQG gain}

$\mc{K}_2$ in Lemma \ref{lemma: sparsity of the LQG gain} corresponds to the first block of $\mc{T}_2-GE^{n}\mc{T}_1^{\dagger} \mc{M}$ in \eqref{eq: static controller in z}. We start by expanding $GE^{n}\mc{T}_1^{\dagger} \mc{M}$. Since $\mc{T}_1^{\dagger}$ is full column rank, we have
\begin{align*}
\mc{T}_1^{\dagger}&=\left(\mc{T}_1^{\T} \mc{T}_1\right)^{-1}\mc{T}_1^{\T}\\
&=\underbrace{\left(G^{\T}G  + \cdots + (E^{n-1})^{\T} G^{\T} GE^{n-1}\right)^{-1}}_{\triangleq S}\mc{T}_1^{\T},
\end{align*}
then we have
\begin{align*}
&G E^n \mc{T}_1^{\dagger}\mc{M}=\\
&GE^n S 
\left[\begin{array}{c;{2pt/2pt}c}
{G^{\T}H+\cdots + (E^{n-1})^{\T} G^{\T} GE^{n-2}F} & \small\text{X}
\end{array}\right],
\end{align*}
where $\text{X}$ denotes any matrix. Then, we take the first block of $G E^n \mc{T}_1^{\dagger}\mc{M}$ and the first block of $\mc{T}_2$ to write $\mc{K}_2$ as
\begin{align*}
\mc{K}_2 =  &GE^{n-1}F \\
&- GE^n S \left(G^{\T}H + \cdots + (E^{n-1})^{\T} G^{\T} GE^{n-2}F\right)\\
\overset{\text{(a)}}{=}& GE^{n-1}F \\
&- GE^n \underbrace{S \left(\overline{G}^{\T}\overline{G} + \cdots + (\overline{E}^{n-1})^{\T} \overline{G}^{\T} \overbar{G}\overline{E}^{n-1}\right)}_{\overset{\text{(b)}}{=} I}\overline{F}\\
\overset{\text{(c)}}{=}& GE^{n-1}F - GE^{n-1}F = 0,
\end{align*}
where in steps (a), (b) and (c) we have used Lemma \ref{lemma: compensator forms}.~\QEDA

\subsection{Proof of Lemma \ref{lemma: number of expert data}}\label{app: number of expert data}
\noindent Since the rank of $Y_N$ in \eqref{eq: expert data} is $\left. \Rank(Y_N) \leq nm+np\right.$, $\left. k=nm+np\right.$ columns are enough for $\Rank(Y_N)$ to stop increasing. To construct $Y_N$ with $\left. k=nm+np\right.$ columns, $ nm+np+n$ samples are required. Therefore, $\left. N=nm+np+n\right.$ expert samples are sufficient to learn the LQG gain $\mc{K}$. This completes the proof.~\QEDA
\end{document}